\documentclass[11pt]{article}
\usepackage{amsmath,amsthm,amssymb,amsfonts}
\newcommand{\field}[1]{\mathbb{#1}}
\newcommand{\remove}[1]{}
\usepackage{longtable}
\usepackage{complexity}

\newtheorem{thm}{Theorem}[section]
\newtheorem{claim}[thm]{Claim}
\newtheorem{lem}[thm]{Lemma}
\newtheorem{define}[thm]{Definition}
\newtheorem{cor}[thm]{Corollary}
\newtheorem{obs}[thm]{Observation}

\newtheorem{fact}[thm]{Fact}
\newtheorem{prop}[thm]{Proposition}
\newtheorem{rmk}[thm]{Remark}

\def\F{{\mathbb{F}}}

\def\Z{{\mathbb{Z}}}

\def\_{\,\,\,\,\,}

\def\spsp{\Sigma\Pi\Sigma\Pi}

\def\eval{\mathsf{Eval}}
\def\dim{\mathsf{Dim}}
\def\lm{\mathsf{Lead\mbox{-}Mon}}
\begin{document}


\mathchardef\mhyphen="2D 

\title{Superpolynomial lower bounds for general homogeneous depth 4 arithmetic circuits}
\author{Mrinal Kumar\thanks{Department of Computer Science, Rutgers University.
Email: \texttt{mrinal.kumar@rutgers.edu}.}\and
Shubhangi Saraf\thanks{Department of Computer Science and Department of Mathematics, Rutgers University.
Email: \texttt{shubhangi.saraf@gmail.com}.}}

\date{}
\maketitle
\abstract{
In this paper, we prove superpolynomial lower bounds for the class of homogeneous depth 4 arithmetic circuits. We give an explicit polynomial in $\VNP$ of degree $n$ in $n^2$ variables such that any homogeneous depth 4 arithmetic circuit computing it must have size $n^{\Omega(\log \log n)}$. 

Our results extend the works of Nisan-Wigderson~\cite{NW95} (which showed superpolynomial lower bounds for homogeneous depth 3 circuits), Gupta-Kamath-Kayal-Saptharishi and Kayal-Saha-Saptharishi~\cite{GKKS12, KSS13} (which showed superpolynomial lower bounds for homogeneous depth 4 circuits with bounded bottom fan-in), Kumar-Saraf~\cite{KS-formula} (which showed superpolynomial lower bounds for homogeneous depth 4 circuits with bounded top fan-in) and Raz-Yehudayoff and Fournier-Limaye-Malod-Srinivasan~\cite{RY08b, FLMS13} (which showed superpolynomial lower bounds for multilinear depth 4 circuits). Several of these results in fact showed exponential lower bounds. 

The main ingredient in our proof is a new complexity measure of {\it bounded support} shifted partial derivatives. This measure allows us to prove exponential lower bounds for homogeneous depth 4 circuits where all the monomials computed at the bottom layer have {\it bounded support} (but possibly unbounded degree/fan-in), strengthening the results of Gupta et al and Kayal et al~\cite{GKKS12, KSS13}. This new lower bound combined with a careful ``random restriction" procedure (that transforms general depth 4 homogeneous circuits to depth 4 circuits with bounded support) gives us our final result. 
}



\section{Introduction}

Proving lower bounds for explicit polynomials is one of the most important open problems in the area of algebraic complexity theory.  Valiant~\cite{Valiant79} defined the classes $\VP$ and $\VNP$ as the algebraic analog of the classes $\P$ and $\NP$, and showed that proving superpolynomial lower bounds for the Permanent would suffice in separating $\VP$ from $\VNP$.  Despite the amount of attention received by the problem, we still do not know any superpolynomial (or even {\it quadratic}) lower bounds for general arithmetic circuits. This absence of progress on the general problem has led to a lot of attention on the problem of proving lower bounds for restricted classes of arithmetic circuits. The hope is that an understanding of restricted classes might lead to a better understanding of the nature of the more general problem, and the techniques developed in this process could possibly be adapted to understand general circuits better. Among the many restricted classes of arithmetic circuits that have been studied with this motivation, {\it bounded depth} circuits have received a lot of attention. 

In a striking result, Valiant et al~\cite{VSBR83} showed that any $n$ variate polynomial of degree $\text{poly}(n)$ which can be computed by a polynomial sized arithmetic circuit of arbitrary depth can also be computed by an arithmetic circuit of depth $O(\log^2 n)$ and size poly$(n)$. Hence, proving superpolynomial lower bounds for circuits of depth $\log^2 n$ is as hard as proving lower bounds for general arithmetic circuits. In a series of recent works, Agrawal-Vinay~\cite{AV08}, Koiran~\cite{koiran} and Tavenas~\cite{Tavenas13} showed that the depth reduction techniques of Valiant et al~\cite{VSBR83} can in fact be extended much further. They essentially showed that in order to prove superpolynomial lower bounds for general arithmetic circuits, it suffices to prove strong enough lower bounds for just {\it homogeneous depth 4} circuits. In particular, to separate $\VNP$ from $\VP$, it would suffice to focus our attention on proving strong enough lower bounds for homogeneous depth 4 circuits. 

The first superpolynomial lower bounds for homogeneous circuits of depth 3 were proved by Nisan and Wigderson~\cite{NW95}. Their main technical tool was the use of the {\it dimension of partial derivatives} of the underlying polynomials as a complexity measure. For many years thereafter, progress on the question of improved lower bounds stalled. In a recent breakthrough result on this problem, Gupta, Kamath, Kayal and Saptharishi~\cite{GKKS12} proved the first superpolynomial ($2^{\Omega(\sqrt n)}$) lower bounds for homogeneous depth 4 circuits when the fan-in of the product gates at the bottom level is bounded (by $\sqrt n$). This result was all the more remarkable in light of the results by Koiran~\cite{koiran} and Tavenas~\cite{Tavenas13} which showed that $2^{\omega(\sqrt n\log n)}$ lower bounds for this model would suffice in separating $\VP$ from $\VNP$. The results of Gupta et al were further improved upon by Kayal Saha and Sapthrashi~\cite{KSS13} who showed $2^{\Omega(\sqrt n\log n)}$ lower bounds for the model of homogeneous depth 4 circuits when the fan-in of the product gates at the bottom level is bounded (by $\sqrt n$). Thus even a slight asymptotic improvement in the exponent of either of these bounds would imply lower bounds for general arithmetic circuits! 

The main tool used in both the papers~\cite{GKKS12} and~\cite{KSS13} was the notion of the dimension of {\it shifted partial derivatives} as a complexity measure, a refinement of the Nisan-Wigderson  complexity measure of dimension of partial derivatives. 

In spite of all this exciting progress on homogeneous depth 4 circuits with bounded bottom fanin (which suggests that possibly we might be within reach of  lower bounds for much more general classes of circuits) these results give almost no non trivial (not even super linear) lower bounds for general homogeneous depth 4 circuits (with no bound on bottom fanin). Indeed the only lower bounds we know for general homogeneous depth 4 circuits are the slightly superlinear  lower bounds by Raz using the notion of elusive functions~\cite{Raz10b}. 

Thus nontrivial lower bounds for the  class of general depth 4 homogeneous circuits seems like a natural and basic question left open by these works, and strong enough lower bounds for this model seems to be an important barrier to overcome before proving lower bounds for more general classes of circuits. 

 In this direction, building upon the work in~\cite{GKKS12, KSS13}, Kumar and Saraf~\cite{KS-depth4, KS-formula} proved superpolynomial lower bounds for depth 4 circuits with unbounded bottom fan-in but {\it bounded top fan-in}. For the case of {\it multilinear} depth 4 circuits, superpolynomial lower bounds were first proved by Raz and Yehudayoff~\cite{RY08b}. These lower bounds were recently improved in a paper by Fournier, Limaye, Malod and Srinivasan~\cite{FLMS13}.  The main technical tool in the work of Fournier et al was the use of the technique of {\it random restrictions} before using shifted partial derivatives as a complexity measure. By setting a large collection of variables at random to zero, all the product gates with high bottom fan-in got set to zero. Thus the resulting circuit had bounded bottom fanin and then known techniques of shifted partial derivatives could be applied. This idea of random restrictions crucially uses the multilinearity of the circuits, since in multilinear circuits high bottom fanin means {\it many} distinct variables feeding in to a gate, and thus if a large collection of  variables is set at random to zero, then with high probability that gate is also set to zero. 



\vspace{2mm}
\noindent
{\bf Our Results: } 
In this paper, we prove the first superpolynomial lower bounds for general homogeneous depth 4 circuits with no restriction on the fan-in, either top or bottom. The main ingredient in our proof is a new complexity measure of {\it bounded support} shifted partial derivatives. This measure allows us to prove exponential lower bounds for homogeneous depth 4 circuits where all the monomials computed at the bottom layer have only few variables (but possibly large degree/fan-in). This exponential lower bound combined with a careful ``random restriction" procedure that allows us to transform general depth 4 homogeneous circuits to this form gives us our final result. We will now formally state our results. 

  Our main theorem is stated below.

\begin{thm}[Lower bounds for homogeneous $\spsp$ circuits]~\label{thm:main}
There is an explicit family of homogeneous polynomials of degree $n$ in $n^2$ variables in $\VNP$ which requires homogeneous $\spsp$ circuits of size $n^{\Omega(\log\log n)}$ to compute it.
\end{thm}

We prove our lower bound for the family of Nisan-Wigderson polynomials $NW_d$ which is based upon the idea of Nisan-Wigderson designs. We give the formal definition in Section~\ref{sec:prelims}. 

As a first step in the proof of Theorem~\ref{thm:main}, we prove an exponential lower bound on the top fan-in of any homogeneous $\spsp$ circuit where every product gate at the bottom level has at most $O(\log n)$ distinct variables feeding into it.  Let homogeneous $\spsp^{\{s\}}$ circuits denote the class of   homogeneous $\spsp$ circuits where every product gate at the bottom level has at most $s$ distinct variables feeding into it (i.e. has support at most $s$). 

\begin{thm}[Lower bounds for homogeneous $\spsp$ circuits with bounded bottom support]~\label{thm:main2}
There exists a constant $\beta > 0$, and an explicit family of homogeneous polynomials of degree $n$ in $n^2$ variables in $\VNP$ such that any homogeneous $\spsp^{\{\beta\log n\}}$ circuit computing it must have top fan-in at least $2^{\Omega(n)}$.
\end{thm}

Observe that since homogeneous $\spsp^{\{s\}}$ circuits are a more general class of circuits than homogeneous $\spsp$ circuits with bottom fan-in at most $s$, our result strengthens the results of of Gupta et al and Kayal et al~\cite{GKKS12, KSS13} when $s = O(\log n)$.

We prove Theorem~\ref{thm:main} by applying carefully chosen random restrictions to both the polynomial family and to any arbitrary homogeneous $\spsp$ circuit and showing that with high probability the circuit simplifies into a homogeneous $\spsp$ circuit with bounded bottom support while the polynomial (even after the restriction) is still rich enough for Theorem~\ref{thm:main2} to hold. 
Our results hold over every field. 

\vspace{2mm}
\noindent
{\bf Organization of the paper :} The rest of the paper is organized as follows.  In Section~\ref{sec:overview}, we provide a high level overview of the proof. In Section~\ref{sec:prelims}, we introduce some notations and preliminary notions used in the paper. In Section~\ref{sec:small-support-lb}, we give a proof of Theorem~\ref{thm:main2}. In Section~\ref{sec:rand-res}, we describe the random restriction procedure and  analyze its effect on the circuit and the polynomial.   In Section~\ref{sec:lowerbounds}, we prove Theorem~\ref{thm:main}. We conclude with some open problems in Section~\ref{sec:conclusion}.


\section{Proof Overview}~\label{sec:overview}
Our proof is divided into two parts. In the first part we show a $2^{\Omega(n)}$ lower bound for homogeneous $\spsp$ circuits whose {\it bottom support} is at most $O(\log n)$. 
To the best of our knowledge, even when the bottom support is $1$, none of the earlier lower bound techniques sufficed for showing nontrivial lower bounds for this model. Thus a new complexity measure was needed. We consider the measure of {\it bounded support} shifted partial derivatives, a refinement of the measure of shifted partial derivatives used in several recent works~\cite{GKKS12,KSS13,KS-depth4, KS-formula, FLMS13}. For this measure, we show that the complexity of the $NW_d$ polynomial (an explicit polynomial in VNP) is {\it high} whereas any subexponential sized homogeneous depth 4 circuit with bounded bottom support has a much smaller complexity measure. Thus for any depth 4 circuit to compute the $NW_d$ polynomial, it must be large -- we show that it must have exponential top fan-in. Thus we get an exponential lower bound for bounded bottom support homogeneous $\spsp$ circuits. We believe this result might be of independent interest. 

In the second part we show how to ``reduce" any $\spsp$ circuit that is not too large to a $\spsp$ circuit with bounded bottom support. This reduction basically follows from a random restriction procedure that sets some of the variables feeding into the circuit to zero. At the same time we ensure that when this random restriction procedure is applied to $NW_d$, the polynomial does not get affected very much, and still has large complexity.  

We could have set variables to zero by picking the variables to set to zero independently at random. For instance consider the following process: Independently keep each variable alive (i.e. nonzero) with probability $1/n^\epsilon$. Then any monomial with $\Omega(\log n)$ distinct variables is set to the zero polynomial with probability at least $1 - 1/n^{\Omega(\log n)}$. Since any circuit of size $n^{o(\log n)}$ will have only $n^{o(\log n)}$ monomials computed at the bottom layer, hence by the union bound, each such monomial with $\Omega(\log n)$ distinct variables will be set to zero. Thus the resulting circuit will have bounded bottom support. The problem with this approach is that we do not know how to analyze the effect of this simple randomized procedure on $NW_d$. Thus we define a slightly more refined random restriction procedure which keeps the $NW_d$ polynomial hard and at the same time makes the $\spsp$ circuit one of bounded bottom support. We describe the details of this procedure in Section~\ref{sec:rralgorithm}


\section{Preliminaries and Notations}~\label{sec:prelims}
\noindent
{\bf Arithmetic Circuits: } An arithmetic circuit over a field $\F$ and a set of variables 
$x_1, x_2, \ldots, x_{N}$ is an directed acyclic graph whose internal nodes are labelled by the field operations and the leaf nodes are labelled by the variables or field elements.  The nodes with fan-out zero are called the output gates and the nodes with fan-in zero are called the leaves. In this paper, we will always assume that there is a unique output gate in the circuit. The {\it size} of the circuit is the number of nodes in the underlying graph and the {\it depth} of the circuit is the length of the longest path from the root to a leaf. We will call a circuit  {\it homogeneous} if the polynomial computed at every node is a homogeneous polynomial. By a $\spsp$ circuit or a depth 4 circuit, we mean a circuit of depth 4 with the top layer and the third layer only have sum gates and the second and the bottom layer have only product gates. In this paper, we will confine ourselves to working with homogeneous depth 4 circuits.  A homogeneous polynomial $P$ of degree $n$ in $N$ variables, which is computed by a homogeneous $\spsp$ circuit can be written as 

\begin{equation}\label{def:model}
P(x_1, x_2, \ldots, x_{N}) = \sum_{i=1}^{T}\prod_{j=1}^{d_i}{Q_{i,j}(x_1, x_2, \ldots, x_{N})}
\end{equation}
Here, $T$ is the top fan-in of the circuit. Since the circuit is homogeneous, we know that for every $i \in \{1, 2, 3, \ldots, T\}$, $$\sum_{j = i}^{d_i} \text{deg}(Q_{i,j}) = n$$
By the support of a monomial $\alpha$, we will refer to the set of variables which have a positive degree in $\alpha$. In this paper, we will also study the class of homogeneous $\spsp$ circuits such that for every $i, j$, every monomial in $Q_{i,j}$ has bounded support. We will now formally define this class. 

\vspace{2mm}
\noindent
{\bf Homogeneous $\Sigma\Pi\Sigma\Pi^{\{s\}}$ Circuits:} A homogeneous $\spsp$ circuit in Equation~\ref{def:model}, is said to be a $\Sigma\Pi\Sigma\Pi^{\{s\}}$ circuit if every product gate at the bottom level has support at most $s$.  Observe that there is no restriction on the bottom fan-in except that implied by the restriction of homogeneity.

\vspace{2mm}
\noindent
{\bf Shifted Partial Derivatives: } In this paper will use a variant of the notion of {\it shifted partial derivatives} which was introduced in~\cite{Kayal12} and has subsequently been the complexity measure used to to prove lower bounds for various restricted classes of depth four circuits and formulas~\cite{FLMS13, GKKS12, KSS13, KS-depth4, KS-formula}.  For a field $\F$, an $N$ variate polynomial $P \in {{{\F}}}[x_1, \ldots, x_{N}]$ and a positive integer $r$, we denote by $\partial^{r} P$, the set of all partial derivatives of order equal to $r$ of $P$. For a polynomial $P$ and a monomial $\gamma$, we denote by ${\partial_{\gamma} (P)}$ the partial derivative of $P$ with respect to $\gamma$. We now reproduce the formal definition from~\cite{GKKS12}.

\begin{define}[Order-$r$ $\ell$-shifted partial derivatives]\label{def:shiftedderivative}
For an $N$ variate polynomial $P \in {\field{F}}[x_1, x_2, \ldots, x_{N}]$ and positive integers $r, \ell \geq 0$, the space of order-$r$ $\ell$-shifted  
partial derivatives of $P$ is defined as
\begin{align}
 \langle \partial^{r} P\rangle_{ \ell} \stackrel{def}{=} \field{F}\mhyphen span\{\prod_{i\in [N]}{x_i}^{j_i}\cdot g  :  \sum_{i\in [N]}j_i = \ell, g \in \partial^{r} P\}
\end{align}
\end{define}

In this paper, we introduce the variation of {\it bounded support} shifted partial derivatives as a complexity measure. The basic difference is that instead of shifting the partial derivatives by all monomials of degree $\ell$, we will shift the partial derivatives only by only those monomials of degree $\ell$ which have support(the number of distinct variables which have non-zero degree in the monomial) exactly equal to $m$. We now formally define the notion. 

\begin{define}[Support-m degree-$\ell$ shifted partial derivatives of order-r]\label{def:restshiftedderivative}
For an $N$ variate polynomial $P \in {\field{F}}[x_1, x_2, \ldots, x_{N}]$ and positive integers $r, \ell, m \geq 0$, the space of support-m degree-$\ell$ shifted  
partial derivatives of order-$r$ of $P$ is defined as
\begin{align}
 \langle \partial^{r} P\rangle_{(\ell, m)} \stackrel{def}{=} \field{F}\mhyphen span\{\prod_{\substack{i\in S \\ S\subseteq [N] \\ |S| = m}}{x_i}^{j_i}\cdot g  :  \sum_{i\in S}j_i = \ell, j_i \geq 1, g \in \partial^{r} P\}
\end{align}
\end{define}

The following property follows from the definition above. 

\begin{lem}~\label{subadditive}
For any two multivariate polynomials $P$ and $Q$ in $\F[x_1, x_2, \ldots, x_{N}]$ and any positive integers $r, \ell, m$, and scalars $\alpha$ and $\beta$
$$\dim(\langle \partial^{r} (\alpha P + \beta Q)\rangle_{(\ell, m)})  \leq \dim(\langle \partial^{r} P\rangle_{(\ell, m)}) + \dim(\langle \partial^{r} Q\rangle_{(\ell, m)})$$
\end{lem}

In the rest of the paper, we will use the term $(m, \ell, r)$-shifted partial derivatives to refer to support-m degree-$\ell$ shifted partial derivatives of order-r of a polynomial. For any linear or affine space $V$ over a field $\F$, we will use $\dim(V)$ to represent the dimension of $V$ over $\F$. We will use the dimension of the space  $\langle\partial^{r} P\rangle_{(\ell, m)}$ which we denote by $\dim(\langle\partial^{r} P\rangle_{(\ell, m)})$ as the measure of complexity of a polynomial. 

\vspace{2mm}
\noindent
{\bf Nisan-Wigderson Polynomials:} We will show our lower bounds for a family of polynomials in $\VNP$ which were used for the first time in the context of lower bounds in~\cite{KSS13}. The construction is based upon the intuition that over any finite field, any two distinct low degree polynomials do not agree at too many points. For the rest of this paper, we will assume $n$ to be of the form $2^k$ for some positive integer $k$. Let $\F_n$ be a field of size $n$. For the set of $N = n^2$ variables $\{x_{i,j} : i, j \in [n]\} $ and $d < n$, we define the degree $n$ homogeneous polynomial $NW_{d}$ as 

$$NW_d = \sum_{\substack{f(z) \in \F_n[z] \\
                        deg(f) \leq d-1}} \prod_{i \in [n]} x_{i,f(i)}$$
                       
From the definition, we can observe the following properties of $NW_d$. 
\begin{enumerate}
\item The number of monomials in $NW_d$ is exactly $n^d$. 
\item Each of the monomials in $NW_d$ is multilinear.
\item Each monomial corresponds to evaluations of a univariate polynomial of degree at most $d-1$ at all points of $\F_n$. Thus, any two distinct monomials agree in at most $d-1$ variables in their support. 
\end{enumerate}

For any $S \subseteq [n]$ and each $f \in \F_n[z]$, we define the monomial $${m_f^S} = \prod_{i\in S} x_{i, f(i)}$$ and $${m_f} = \prod_{i\in[n]} x_{i, f(i)}$$ We also define the set ${\cal M}^S$ to represent the set $\{\prod_{i \in S}\prod_{j \in [n]} x_{i,j}\}$.  Clearly, 
$$NW_d = \sum_{\substack{f(z) \in \F_n[z] \\
                        deg(f) \leq d-1}} m_f $$

\vspace{2mm}
\noindent
{\bf Monomial Ordering and Distance: }
We will also use the notion of a monomial being an extension of another as defined below. 
\begin{define}
A monomial $\theta$ is said to be an extension of a monomial $\tilde{\theta}$, if $\theta$ divides $\tilde{\theta}$. 
\end{define}

\noindent
In this paper, we will imagine our variables to be coming from a $n\times n$ matrix $\{x_{i,j}\}_{i, j \in [n]}$. We will also consider the following total order on the variables. $x_{i_1, j_1} > x_{i_2, j_2}$ if either $i_1 < i_2$ or $i_1 = i_2$ and $j_1 < j_2$. This total order induces a lexicographic order on the monomials. For a polynomial $P$,  we will use the notation $\lm(P)$ to indicate the leading monomial of $P$ under this monomial ordering.

We will use the following notion of distance between two monomials which was also used in~\cite{CM13}. 
\begin{define}[Monomial distance]
Let $m_1$ and $m_2$ be two monomials over a set of variables. Let $S_1$ and $S_2$ be the multiset of variables in $m_1$ and $m_2$ respectively, then the distance $\Delta(m_1, m_2)$ between $m_1$ and $m_2$ is the min$\{|S_1| - |S_1\cap S_2|, |S_2| - |S_1\cap S_2|\}$ where the cardinalities are the order of the multisets.   
\end{define} 

In this paper, we will invoke this definition only for multilinear monomials of the same degree. In this special case, we have the following crucial observation.

\begin{obs}~\label{obs:multilinear-dist} 
Let $\alpha$ and $\beta$ be two multilinear monomials of the same degree which are at a distance $\Delta$ from each other. If $\text{Supp}(\alpha)$ and $\text{Supp}(\beta)$ are the supports of $\alpha$ and $\beta$ respectively, then $$|\text{Supp}(\alpha)| - |\text{Supp}(\alpha)\cap \text{Supp}(\beta)| =  |\text{Supp}(\beta)| - |\text{Supp}(\alpha)\cap \text{Supp}(\beta)| =  \Delta$$
\end{obs}

\vspace{2mm}
\noindent
{\bf Approximations: } We will repeatedly refer to the following lemma to approximate expressions during our calculations. 

\begin{lem}[\cite{GKKS12}]~\label{lem:approx}
Let $a(n), f(n), g(n) : \Z_{>0}\rightarrow \Z_{>0}$ be integer valued functions such that $(f+g) = o(a)$. Then,
$$\log \frac{(a+f)!}{(a-g)!} = (f+g)\log a \pm O\left( \frac{(f+g)^2}{a}\right)$$
\end{lem} 



In our setup, very often $(f+g)^2$ will  be $\theta(a)$. In this case, the error term will be an absolute constant. Hence, up to multiplication by constants, $\frac{(a+f)!}{(a-g)!} = a^{(f+g)}$.






We will also use the following basic fact in our proof. 

\begin{fact}~\label{fact:numsolutions}
  The number of {\it positive} integral solutions of the equation 
$$\sum_{i = 1}^t y_i = k $$ equals  ${k-1 \choose t-1}$.
\end{fact}

As a last piece of notation, for any  $i\times j$ matrix $H$ over $\F_2$ and a vector $\alpha \in \F^{i}_2$, we denote by $H||\alpha$ to be the $i\times (j+1)$ matrix which when restricted to the first $j$ columns is equal to $H$ and whose last column is $\alpha$. Similarly, for any vector $\alpha \in \F^i_2$ and any $b\in\F_2$, $\alpha||b$ is the $i+1$ dimensional vector where $b$ is appended to $\alpha$.

\noindent



\section{Lower bounds for $\spsp^{\{O(\log n)\}}$ circuits}~\label{sec:small-support-lb}
In this section, we will prove Theorem~\ref{thm:main2}. We will prove an exponential lower bound on the top fan-in for homogeneous $\spsp$ circuits such that every product gate at the bottom has a bounded number of variables feeding into it. We will use the dimension of the span of $(m, \ell, r)$-shifted partial derivatives as the complexity measure. We will prove our lower bound for the $NW_d$ polynomial. The proof will be in two parts. In the first part, we will prove an upper bounded on the complexity of the circuit. Then, we will prove a lower bound on the complexity of the $NW_d$ polynomial. Comparing the two  will then imply our lower bound. The bound holds for $NW_d$ for any $d = \delta n$, where $\delta$ is a constant such that $0 < \delta < 1$.

\subsection{Complexity of homogeneous depth 4 $\spsp^{\{s\}}$ circuits}
Let $C$ be a homogeneous $\spsp^{\{s\}}$ circuit computing the $NW_d$ polynomial. We will now prove an upper bound on the complexity of a product gate in such a circuit. The bound on the complexity of the circuit follows from the subadditivity of the complexity measure. 

\begin{lem}~\label{lem:product-gate-bound1}
Let $Q = \prod_{i = 1}^n Q_i$ be a product gate at the second layer from the top in a homogeneous $\spsp^{\{s\}}$ circuit computing a homogeneous degree $n$ polynomial in $N$ variables.  For any positive integers $m, r, s, \ell$ satisfying $m+rs \leq \frac{N}{2}$ and $m+rs \leq \frac{\ell}{2}$, 

$$\dim(\langle\partial^{r} Q\rangle_{(\ell, m)}) \leq \text{poly}(nrs){n+r \choose r}{N \choose m+rs}{\ell+n-r \choose m+rs}$$
\end{lem}

\begin{proof}
By the application of chain rule, any partial derivative of order $r$ of $Q$ is a linear combination of a number of product terms. Each of these product terms is of the form $\prod_{i \in S}\partial_{\gamma_i}(Q_i)\prod_{j\in [n]\setminus S}Q_j$, where $S$ is a subset of $\{1, 2, \ldots, n\}$ of size at most $r$ and $\gamma_i$ are monomials such that $\sum_{i \in S}\text{deg}(\gamma_i) = r$. Also, observe that $\prod_{i \in S}\partial_{\gamma_i}(Q_i)$ is of degree at most $n-r$. In this particular special case all $Q_i$ have support at most $s$, so every monomial in $\prod_{i \in S}\partial_{\gamma_i}(Q_i)$ has support at most $rs$. Shifting  these derivatives is the same as multiplying them with monomials of degree $\ell$ and support equal to $m$. So, $(m, \ell, r)$-shifted partial derivative of order $r$ can be expressed as sum of the  product of $\prod_{j\in [n]\setminus S}Q_j$ for $S\subseteq [n]$ of size at most $r$, and a monomial of support between $m$ and $m+rs$ and degree between $\ell$ and $\ell+n-r$. 

We can choose the set $S$ in ${n+r \choose r}$ ways. The second part in each term is a monomial of degree between $l$ and $\ell+n-r$ and support between $m$ and $m+rs$. The number of monomials over $N$ variables of support between $m$ and $m+rs$ and degree between $\ell$ and $\ell+n-r$  equals 
$$\sum_{i = 0}^{n-r} {\sum_{j = 0}^{rs} {N \choose m+j}{\ell+i -1\choose m+j-1}}$$ 
Now, in the range of choice of our parameters $m, r, s, \ell$, the binomial coefficients increase monotonically with $i$ and $j$. Hence, we can upper bound the dimension by $\text{poly}(nrs){n+r \choose r}{N \choose m+rs}{\ell+n-r-1 \choose m+rs-1}$. 
\end{proof}

For a homogeneous $\spsp$ circuit where each of the bottom level product gates is of support at most $s$, Lemma~\ref{lem:product-gate-bound1} immediately implies the following upper bound on the complexity of the circuit due to subadditivity from Lemma~\ref{subadditive}. 
\begin{cor}[Upper bound on circuit complexity]~\label{cor:circuit-complexity-bound}
Let $C = \sum_{j = 1}^{T}\prod_{i = 1}^n Q_{i,j}$ be a  a homogeneous $\spsp^{\{s\}}$ circuit computing a homogeneous degree $n$ polynomial in $N$ variables. For any $m, r, s, \ell$ satisfying $m+rs \leq \frac{N}{2}$ and $m+rs \leq \frac{\ell}{2}$, 
$$\dim(\langle\partial^{r} C\rangle_{(\ell, m)}) \leq T\times \text{poly}(nrs){n+r \choose r}{N \choose m+rs}{\ell+n-r-1 \choose m+rs-1}$$
\end{cor}

\subsection{Lower bound on the complexity of the $NW_d$ polynomial}
We will now prove a lower bound on the complexity of the $NW_d$ polynomial. For this, we will first observe that distinct partial derivatives of the $NW_d$ polynomial are {\it far} from each other in some sense and then show that shifting such partial derivatives gives us a lot of distinct shifted partial derivatives. Recall that we defined the set ${\cal M}^S$ to represent the set $\{\prod_{i \in S}\prod_{j \in [n]} x_{i,j}\}$.
We start with the following observation.

\begin{lem}~\label{lem:many-partial-derivatives}
For any positive integer $r$ such that $n-r > d$ and $r < d-1$, the set $\{\partial_{\alpha}(NW_d) : \alpha \in {\cal M}^{[r]}\}$ consists of $|{\cal M}^{[r]}| = n^r$ nonzero distinct polynomials. 
\end{lem}
\begin{proof}
We need to show the following two statements. 
\begin{itemize}
\item $\forall \alpha \in {\cal M}^{[r]}$, $\partial_{\alpha}(NW_d)$ is a non zero polynomial. 
\item $\forall \alpha \neq \beta \in {\cal M}^{[r]}$, $\partial_{\alpha}(NW_d) \neq \partial_{\beta}(NW_d)$.
\end{itemize}
For the first item,  observe that, since $r < d-1$, for every $\alpha \in {\cal M}^{[r]}$, there is a polynomial $f$ of degree at most $d-1$ in $\F_n[z]$ such that $\alpha = \prod_{i = 1}^r x_{i, f(i)}$. So, $\partial_{\alpha}(m_f) \neq 0$ since $m_f$ is an extension of $\alpha$, in fact, there are many such extensions. Also, observe for any two extensions $m_f$ and $m_g$, $\partial_{\alpha}(m_f)$ and $\partial_{\alpha}(m_g)$ are multilinear monomials at a distance at least $n-r-d > 0$ from each other. Hence, $\partial_{\alpha}(NW_d) = \sum_{g}  \partial_{\alpha}(m_g)$ is a non zero polynomial, where the sum is over all $g \in \F_n[z]$ of degree $\leq d-1$ such that $m_g$ is an extension of $\alpha$. 

For the second item, let us now consider the leading monomials of $\partial_{\alpha}(NW_d)$ and $\partial_{\beta}(NW_d)$. These leading monomials each come from some distinct polynomials $f, g \in \F_n[z]$ of degree at most $d-1$. Also, since $\alpha \neq \beta$ and $n-r > d$,  $\partial_{\alpha}(m_f) \neq \partial_{\beta}(m_g)$. In fact, $\partial_{\alpha}(NW_d) \text{ and } \partial_{\beta}(NW_d)$ do not have a common monomial. Therefore, $\partial_{\alpha}(NW_d) \neq \partial_{\beta}(NW_d)$.  
\end{proof}

\begin{rmk}~\label{rmk2}
Observe that there is nothing special about the set ${\cal M}^{[r]}$ and the Lemma~\ref{lem:many-partial-derivatives} holds for $\{\cal M\}^S$ for any set $S$, such that $S \subseteq [n]$ and $|S| < d-1$.
\end{rmk}

In the proof above, we observed that for any $\alpha \neq \beta \in {\cal M}^{[r]}$, the leading monomials of $\partial_{\alpha}(NW_d)$ and $ \partial_{\beta}(NW_d)$ are multilinear monomials of at a distance at least $n-r-d$ from each other. We will exploit this structure to show that shifting the polynomials in the set $\{\partial_{\alpha}(NW_d) : \alpha \in {\cal M}^{[r]}\}$ by monomials of support m and degree $\ell$ results in many linearly independent  shifted partial derivatives. We will first prove the following lemma.

\begin{lem}~\label{clm:manyshifts1}
Let $\alpha$ and $\beta$ be two distinct multilinear monomials of equal degree such that the distance between them is  $\Delta$. Let $S_{\alpha}$ and  $S_{\beta}$ be the set of all monomials obtained by shifting $\alpha$ and $\beta$ respectively with monomials of degree $\ell$ and support exactly $m$ over $N$ variables. Then $|S_{\alpha}\cap S_{\beta}| \leq {N-\Delta \choose m-\Delta}{\ell-1 \choose m-1}$. 
\end{lem}

\begin{proof}
From the distance property, we know that there is a unique monomial $\gamma$ of degree $\Delta$ and support $\Delta$  such that $\alpha\gamma$ is the lowest degree extension of $\alpha$ which is divisible by $\beta$.  Therefore, any extension of $\alpha$ which is also an extension of $\beta$ must have the support of $\alpha\gamma$ as a subset. In particular, for a shift of $\alpha$ to lie in $S_{\beta}$, $\alpha$ must be shifted by monomial of degree $\ell$ and support $m$ which is an extension of  $\gamma$. Hence, the freedom in picking the support is restricted to picking some $m-\Delta$ variables from the remaining $N - \Delta$ variables. Once the support is chosen, the number of possible degree $\ell$ shifts on this support equals ${\ell-1 \choose m-1}$ by Fact~\ref{fact:numsolutions}. Hence, the number of shifts of degree equal to $\ell$ and support equal to m of $\alpha$ which equals some degree $\ell$ and support m shift of $\beta$ is exactly ${N-\Delta \choose m-\Delta}{\ell-1 \choose m-1}$. 
\end{proof}

We will now prove the following lemma, which is essentially an application of Claim~\ref{clm:manyshifts1} to the $NW_d$ polynomial. For any monomial $\alpha$ and positive integers $\ell, m$, we will denote by $S_{\ell, m}(\alpha)$ the set of all shifts of $\partial_{\alpha}NW_d$ by monomials of degree $\ell$ and support m. More formally, 
$$S_{\ell, m}(\alpha) = \{\gamma\cdot\partial_{\alpha}(NW_d) : \gamma = \prod_{\substack{i\in U \\ U\subseteq [N] \\ |U| = m}}{x_i}^{j_i} , \sum_{i\in U}j_i = \ell, j_i \geq 1\}$$
also, let $$LM_{\ell, m}(\alpha) = \{{\lm}(f) : f \in S_{\ell, m}(\alpha)\}$$ 
\begin{lem}~\label{lem:manyshifts2}
For any positive integers $r$, $m$ and $\ell$ such that $n-r > d$ and $r < d-1$, let $\alpha$ and $\beta$ be two distinct monomials in ${\cal M}^{[r]}$. Then $|S_{\ell, m}(\alpha)\cap S_{\ell, m}(\beta)| \leq {N-(n-d-r) \choose m- (n-d-r)}{\ell-1 \choose m-1}$. 
\end{lem}

\begin{proof}
In the proof of Lemma~\ref{lem:many-partial-derivatives}, we have observed that the leading monomials of $\partial_{\alpha}(NW_d)$ and $\partial_{\beta}(NW_d)$ are equal to $\partial_{\alpha}(m_f)$ and $\partial_{\beta}(m_g)$ for two distinct polynomials $f, g \in \F_n[z]$ of degree at most $d-1$. Hence, $\partial_{\alpha}(m_f)$ and $\partial_{\beta}(m_g)$ are multilinear monomials at a distance at least $\Delta = n-r-d$ from each other. 

Since monomial orderings respect multiplication by the same polynomial,  we know that the leading monomial of a shift equals the shift of the leading monomial. Therefore, if $\gamma_{\alpha}$ and $\gamma_{\beta}$ are two monomials of degree $\ell$ and support equal to $m$ such that $\gamma_{\alpha}\partial_{\alpha}(NW_d) = \gamma_{\beta}\partial_{\beta}(NW_d)$, then $\gamma_{\alpha}\partial_{\alpha}(m_f) = \gamma_{\beta}\partial_{\beta}(m_g)$. Hence, the $|S_{\ell, m}(\alpha)\cap S_{\ell, m}(\beta)|$ is at most the number of shifts of  $\partial_{\alpha}(m_f)$ which is also a shift of  $\partial_{\beta}(m_g)$. By Lemma~\ref{clm:manyshifts1}, this is at most ${N-(n-d-r) \choose m- (n-d-r)}{\ell-1 \choose m-1}$.
\end{proof}

We will now prove a lower bound on the dimension of the span of $(m, \ell, r)$-shifted partial derivatives of the $NW_d$ polynomial. For this, we will use the following proposition  from~\cite{GKKS12}, the proof of which is a simple application of Gaussian elimination. 

\begin{prop}[\cite{GKKS12}]
For any field $\F$, let ${\cal P} \subseteq \F[z]$ be any finite set of polynomials. Then, 
$$\dim(\F\mbox{-}span({\cal P})) = |\{\lm(f) : f \in \F\mbox{-}span({\cal P})\}| $$
\end{prop}

Therefore, in order to lower bound  $\dim(\langle\partial^{r} NW_d\rangle_{(\ell, m)})$, it would suffice to obtain a lower bound on the size of the set $\bigcup_{\alpha} LM_{ \ell, m}(\alpha)$, where the union is over all monomials $\alpha$ of degree equal to $r$. To obtain this lower bound, we will show a lower bound on the size of the set  $\bigcup_{\alpha \in {\cal M}^{[r]}} LM_{\ell, m}(\alpha)$.

\begin{lem}~\label{lem:manyshiftsNW}
Let $d = \delta n$ for any constant $0<\delta<1$. Let $\ell, m, r$ be positive integers such that $n-r > d$, $r < d-1$, $m \leq N$, $m = \theta(N)$ and  for $\phi = {\frac{N}{m}}$,  $r$ satisfies $ r \leq \frac{ (n-d)\log{ \phi} \pm O(\phi\frac{(n-d-r)^2}{N})}{\log n + \log{ \phi}}$. Then, $$\dim(\langle\partial^{r} NW_d\rangle_{(\ell, m)}) \geq 0.5n^r{N \choose m}{\ell-1 \choose m-1}$$
\end{lem}

\begin{proof}
Recall that ${\cal M}^{[r]} = \{\prod_{i = 1}^r\prod_{j \in [n]} x_{i,j}\}$. We have argued in Lemma~\ref{lem:many-partial-derivatives} that for each $\alpha, \beta \in {\cal M}^{[r]}$, such that $\alpha \neq \beta$, $\partial_{\alpha}(NW_d) \neq \partial_{\beta}(NW_d)$ and both of these are non zero polynomials. As discussed above,  we will prove a lower bound on the size of the set $\bigcup_{\alpha \in {\cal M}^{[r]}} LM_{\ell, m}(\alpha)$.  From the principle of inclusion-exclusion, we know 
$$|\bigcup_{\alpha \in {\cal M}^{[r]}} LM_{ \ell, m}(\alpha)| \geq \sum_{\alpha \in {\cal M}^{[r]}}|LM_{\ell, m}(\alpha)| -\sum_{\alpha \neq \beta \in {\cal M}^{[r]}} |LM_{ \ell, m}(\alpha) \cap LM_{\ell, m}(\beta)|$$  
Let us now bound both these terms separately. 
\begin{itemize}
\item Since shifting preserves monomial orderings, therefore for any $\gamma \neq \tilde{\gamma}$ of degree $\ell$ and support $m$, and for any $\alpha \in {\cal M}^{[r]}, \lm(\gamma\partial_{\alpha}(NW_d)) \neq \lm(\tilde{\gamma}\partial_{\alpha}(NW_d))$. Hence, for each $\alpha \in {\cal M}^{[r]}$, $|LM_{ \ell, m}(\alpha)|$ is the number of different shifts possible, which is equal to the number of distinct monomials of degree $\ell$ and support $m$ over $N$ variables. Hence, $$|LM_{ \ell, m}(\alpha)| = {N \choose m}{\ell-1 \choose m-1}$$.
\item For any two distinct $\alpha, \beta \in {\cal M}^{[r]}$, from Lemma~\ref{lem:manyshifts2}, $$|LM_{\ell, m}(\alpha) \cap LM_{\ell, m}(\beta)| \leq {N-(n-d-r) \choose m- (n-d-r)}{\ell-1 \choose m-1}$$ 

\end{itemize}

Therefore, 
$$|\bigcup_{\alpha \in {\cal M}^{[r]}} LM_{ \ell, m}(\alpha)| \geq |{\cal M}^{[r]}|{N \choose m}{\ell-1 \choose m-1} - {|{\cal M}^{[r]}| \choose 2}{N-(n-d-r) \choose m- (n-d-r)}{\ell-1 \choose m-1} $$  
To simplify this bound, we will show that for the choice of our parameters, the second term is at most the half the first term. In this case, we have 
$$|\bigcup_{\alpha \in {\cal M}^{[r]}} LM_{\ell, m}(\alpha)| \geq 0.5|{\cal M}^{[r]}|{N \choose m}{\ell-1 \choose m-1}$$

We need to ensure, 
 $$\frac{{|{\cal M}^{[r]}| \choose 2}{N-(n-d-r) \choose m- (n-d-r)}{\ell-1 \choose m-1}}{|{\cal M}^{[r]}|{N \choose m}{\ell-1 \choose m-1} } \leq 0.5$$
 It suffices to ensure

$$ \frac{|{\cal M}^{[r]}|{N-(n-d-r) \choose m- (n-d-r)}}{{N \choose m}} \leq 1 $$
which is the same as ensuring that $$|{\cal M}^{[r]}|\times \frac{(N-(n-d-r))!}{N!}\times \frac{m!}{(m-(n-d-r))!} \leq 1$$

Now, using the  approximation from Lemma~\ref{lem:approx},
\begin{eqnarray*}
\log\frac{N!}{(N-(n-d-r))!} &=& (n-d-r)\log N \pm O\left(\frac{(n-d-r)^2}{N}\right) \text{and} \\
\log\frac{m!}{(m-(n-d-r))!} &=& (n-d-r)\log m \pm O\left(\frac{(n-d-r)^2}{m}\right)
\end{eqnarray*}

Thus we need to ensure that

$$\log |{\cal M}^{[r]}| \leq \log\left(\frac{N}{m}\right)^{n-d-r} \pm O\left(\frac{(n-d-r)^2}{N}\right) \pm O\left(\frac{(n-d-r)^2}{m}\right)$$
 Substituting $|{\cal M}^{[r]}| = n^r$, we need
$$r\log n \leq \log\left(\frac{N}{m}\right)^{n-d-r} \pm  O\left(\frac{(n-d-r)^2}{N} + \frac{(n-d-r)^2}{m}\right)$$ 

Substituting $m = \frac{N}{\phi}$ (and noting that $\phi > 1$),  we require 
 $$r\log n \leq (n-d-r)\log{\phi} \pm  O\left(\phi\frac{(n-d-r)^2}{N}\right).  $$
 Thus we require
 $$r \leq \frac{ (n-d)\log{ \phi} \pm  O(\phi\frac{(n-d-r)^2}{N})}{\log n + \log{ \phi}}$$
 
Observe that for any constant $0< \delta < 1$ such that $d = \delta n$, $r$ can be chosen any constant times $\frac{n}{\log n}$ by choosing $\phi$ to be an appropriately large constant. 
So, for such a choice of $r$, $$\dim(\langle\partial^{r} NW_d\rangle_{(\ell, m)}) \geq 0.5|{\cal M}^{[r]}|{N \choose m}{\ell-1 \choose m-1}$$
For $|{\cal M}^{[r]}| = n^r$, we have $$\dim(\langle\partial^{r} NW_d\rangle_{(\ell, m)}) \geq 0.5n^r{N \choose m}{\ell-1 \choose m-1}$$
\end{proof}

\begin{rmk}~\label{rmk:bounded-sup-lb}
The proof above shows something slightly more general than  a lower bound on just the complexity of the $NW_d$ polynomial. The only property of the $NW_d$ polynomial that we used here was that the leading monomials of any two distinct partial derivatives of it were {\it far} from each other. We will crucially use this observation in the proof of our main theorem. Also, there is nothing special about using the  set ${\cal M}^{[r]}$. The proof works for any set of monomials ${\cal M}^{S} = \{\prod_{i \in S}\prod_{j \in [n]} x_{i,j}\}$, where $S$ is a subset of $\{1, 2, 3,\ldots, n\}$ of size exactly $r$. 
\end{rmk}

\subsection{Top fan-in lower bound}
We are now ready to prove our lower bound on the top fan-in of any homogeneous $\spsp^{\{\beta\log n\}}$ (for some constant $\beta$) and computes the $NW_d$ polynomial, where $d = \delta n$ for some constant $\delta$ between $0$ and $1$. 

\begin{thm}~\label{thm:bounded-support-lower-bound}
Let $d = \delta n$ for any constant $0<\delta<1$. There exists a constant $\beta$ such that all homogeneous $\spsp^{\{\beta\log n\}}$ circuits which compute the $NW_d$ polynomial have top fan-in at least $2^{\Omega(n)}$.
\end{thm}
\begin{proof}
By comparing the complexities of the circuit and the polynomial as given by Corollary~\ref{cor:circuit-complexity-bound} and Lemma~\ref{lem:manyshiftsNW}, the top fan-in of the circuit must be at least 

\begin{equation}~\label{eqn:top-fan-in-1}
\frac{0.5n^{r}{N\choose m}{\ell-1 \choose m-1}}{\text{poly}(nrs){n+r \choose r}{N \choose m+rs}{\ell+n-r \choose m+rs}}
\end{equation}

This bound holds for any choice of positive integers $\ell, m, r$, a constant $\beta$ such that $s = \beta\log n$ which satisfy the constraints in the hypothesis of Corollary~\ref{cor:circuit-complexity-bound} and Lemma~\ref{lem:manyshiftsNW}. In other words, we want these parameters to satisfy 

\begin{itemize}
\item $m+rs \leq \frac{N}{2}$ 
\item $m+rs \leq \frac{\ell}{2}$
\item $m = \theta(N)$
\item $n-r > d$
\item $r < d-1$
\item For $\phi = {\frac{N}{m}}$,  $ r \leq \frac{ (n-d)\log{ \phi} \pm O\left(\phi \frac{(n-d-r)^2}{N}\right)}{\log n + \log{ \phi}}$
\end{itemize}

In the rest of the proof, we will show that there exists a choice of these parameters such that we get a bound of $2^{\Omega(n)}$ from Expression~\ref{eqn:top-fan-in-1}. We will show the existence of such parameters satisfying the asymptotics $\ell = \theta(N)$, $r = \theta\left( \frac{n}{\log n} \right)$ and $s = \theta(\log n)$. In the rest of the proof, we will crucially use these asymptotic bounds for various approximations.



For this, we will group together and approximate the terms in the ratio $\frac{0.5n^{r}{N\choose m}{\ell-1 \choose m-1}}{\text{poly}(nrs){n+r \choose r}{N \choose m+rs}{\ell+n-r \choose m+rs}}$
\begin{itemize}
\item $\frac{{N \choose m}}{{N \choose m+rs}} = \frac{(N-m-rs)!(m+rs)!}{(N-m)!m!} = (\frac{m}{N-m})^{rs}$ upto some constant factors, as long as $(rs)^2 = \theta (N) = \theta(m)$. 
\item $\frac{{\ell-1 \choose m-1}}{{\ell+n-r \choose m+rs}} = {\frac{(\ell-1)!}{(m-1)!(\ell-m)!} \times \frac{(m+rs)!(\ell-m+n-r-rs)!}{(\ell+n-r)!}}$. We now pair up things we know how to approximate within constant factors. $\frac{{\ell-1 \choose m-1}}{{\ell+n-r \choose m+rs}} = \frac{(\ell-1)!}{(\ell+n-r)} \times \frac{(m+rs)!}{(m-1)!} \times \frac{(\ell-m+n-r-rs)!}{(\ell-m)!} = \text{poly(n)} \times {\frac{1}{\ell^{n-r}}}\times m^{rs} \times {\frac{(\ell-m)^{n-r}}{(\ell-m)^{rs}}}$. This simplifies to $\text{poly(n)} \times {\left(\frac{m}{\ell-m}\right)}^{rs} \times {\left(\frac{\ell-m}{\ell}\right)}^{n-r}$.
\item $\frac{n^{r}}{{n+r \choose r}} \geq \frac{n^{r}}{{\left(\frac{2(n+r)}{r}\right)}^{r}}$. We just used Stirling's approximation here. 

\end{itemize}

In the range of our parameters, the approximations above imply that the top fan-in, up to polynomial factors is at least 

$${\left(\frac{r}{3}\right)}^r\times{\left(\frac{m}{\ell-m}\right)}^{rs} \times {\left(\frac{\ell-m}{\ell}\right)}^{n-r} \times \left(\frac{m}{N-m}\right)^{rs}$$

Simplifying further, this is at least

$$2^{\Omega(r\log r - rs\log\frac{\ell-m}{m} - (n-r)\log\frac{\ell}{\ell-m}-rs\log\frac{N-m}{m})}$$

Recall that we will set $m$ and $\ell$ to be $\theta(N)$ and $r$ to be $\theta(\frac{n}{\log n})$. The constants have to be chosen carefully in order to satisfy the constraints. We will choose constants $\alpha, \beta$ and $\eta$ such that $s = \beta\log n$, $r = \alpha\cdot n/\log n$ and $m= \eta \ell$.
First choose $\eta$ to be any small constant $> 0$ (for instance $\eta= 1/4$). Now,  choose $\alpha$ to be a constant much larger than $\log\frac{1}{1-\eta}$. This makes sure that $r\log r$ dominates $(n-r)\log\frac{\ell}{\ell-m}$. Recall that $\alpha$ can be chosen to be any large constant by choosing $\phi$ to be an appropriately large constant (by the constraint between $r$ and $\phi$ in the fifth bullet). Notice that this sets $m$ to be a small constant factor of $N$. Fix these choices of $\eta$ and $\alpha$. Now, we choose the term $\beta$ to be a small positive constant such that $rs\log\frac{1-\eta}{\eta}$ and $rs\log\frac{N-m}{m}$ are much less than $r\log r$. Observe that this choice of parameters satisfies all the constraints imposed in the calculations above, and the top fan-in is at least $2^{\Omega(r\log r)} = 2^{\Omega(n)}$.  
\end{proof}


\section{Random Restrictions}~\label{sec:rand-res}

In this section, we will describe our random restriction algorithm and analyze the effect of random restrictions on $\spsp$ circuits as well as the $NW_d$ polynomial. 

Let $n = 2^k$. We identify elements of $[n]$ with elements of $\F_{2^k}$. We view $\F_{2^k}$ as a $k$-dimensional vector space over $\F_2$. Let $\phi:  \F_{2^k} \to \F_2^k$ be an $\F_2$-linear isomorphism between $\F_{2^k}$ and $\F_2^k$.
Thus $\phi(\alpha + \beta) = \phi(\alpha) + \phi(\beta)$. Let  $M: \F_{2^k} \to \F_2^{k\times k}$, map $\alpha \in  \F_{2^k}$ to the matrix $M(\alpha)$, which represents the linear transformation over $\F_2^k$
that is given by multiplication by $\alpha$ in $\F_{2^k}$.  Thus it follows that $M(\alpha \times \beta) = M(\alpha) \times M(\beta)$, and $M(\alpha + \beta) = M(\alpha) + M(\beta)$. 
Moreover it is not hard to see that $\phi(\alpha \times \beta) = M(\alpha) \times \phi(\beta)$. 

Since $n = 2^k$, thus  $\F_n \equiv \F_{2^k}$. Let $\F_n[Z]$ denote the space of univariate polynomials over $\F_n$. For $f\in \F_n[Z]$ of degree $\leq d-1$, $f$ is of the form  $\sum_{i = 0}^{d-1} a_i Z^i$, for $a_i \in \F_n$. Thus we can represent $f$ as a vector of coefficients $(a_0,a_1, \ldots a_{d-1})$, and hence view $f$ as an element of $\F_n^d$.  For ease of notation, for $\alpha \in \F_n$ we will let $[\alpha]$ represent $\phi(\alpha)$. 
Also, for $f \in \F_n[Z]$ of degree at most $d-1$, we let $[f] \in \F_2^{kd}$ represent the concatenation of $\phi$ applied to each of the coefficients of $f$. 

Let  $\eval_\alpha$ be the $dk \times k$ matrix obtained by stacking the matrices $M(\alpha^0)$, $M(\alpha^1)$, ..., $M(\alpha^{d-1})$ one below the other. In other words, the first $k$ rows are the rows of 
$M(\alpha^0)$, the second $k$ rows are the rows of $M(\alpha^1)$ and so on.  The following claim follows easily from the definitions. 

\begin{claim}
Let $f \in  \F_n[Z]$ be of degree at most $d-1$, and let $\alpha \in \F_n$. Then $$[f(\alpha)] = [f] \times \eval_\alpha.$$
\end{claim}

In the rest of the discussion we will identify the elements of $\F_n$ with $\{1, 2, \ldots, n\}$. 
Let $\overline {\eval_i}$ be the  $dk \times 2^k$ matrix obtained by adding a column for each of the $2^k$ linear combinations of the columns of $\eval_i$. 
Let $\eval$ be the $dk \times nk$ matrix obtained by concatenating $\eval_{i}$ for all $i \in [n]$. 
Let $\overline {\eval}$ be the $dk \times n2^k$ matrix obtained by concatenating  $\overline {\eval_{i}}$ for all $i \in [n]$. 

In order to restrict the variables in the circuit, we  will first ``randomly restrict" the space of polynomials in $\F_n[Z]$ of degree at most $d-1$. We present the random restriction procedure in the next section.

\subsection{Random Restriction Algorithm}\label{sec:rralgorithm}

Let $\epsilon > 0$ be any constant. We will define a randomized procedure $R_\epsilon$ which selects a subset of the variables $\{x_{i,j}\mid i,j \in [n]\}$ to set to zero. 

The restriction proceeds by first restricting the space of polynomials $f \in \F_n[Z]$ of degree at most $d-1$. This restriction then naturally induces a restriction on the space of variables by selecting only those variables $x_{i,j}$ such that there is some polynomial $f$ in the restricted space for which $f(i) = j$. 

We restrict the space of polynomials by iteratively restricting the values the polynomials can take at points in $\F_{2^k}$. For each $i \in F_{2^k}$, we restrict the values $f$ can take at $i$ to a random affine subspace of codimension $\epsilon k$ (when we view $\F_{2^k}$ as a $k$ dimensional vector space over $\F_2$). We do this by sampling $\epsilon k$ random and independent columns from $\overline{\eval_i}$ and restricting the inner product of $[f]$ with these columns to be randomly chosen values. Each column that we pick in this manner imposes an $\F_2$-affine constraint on $[f]$, and restricts $[f]$ to vary in an affine subspace of codimension $1$. Since these random constraints for the various values of $i$ might not be linearly independent, it is possible that at the end of the process no polynomial $f$ satisfies the constraints. Thus we need to be more careful. We iteratively impose these random constraints for various values of $i$, but at the same time ensure that each new constraint that is imposed on $f$ is linearly independent of the old constraints. We do this by making sure that each new column that is sampled is linearly independent of the old columns.

\vspace{5mm}
\noindent{\bf Random restriction procedure $R_\epsilon$}\\
\noindent{\bf Output:} The set of variables that are set to zero. 
\begin{enumerate}
\item Initialize $A_0 = \F_2^{kd}$, $\cal B$ to be a $0$ dimensional vector, $\cal M$ to be an empty matrix over $\F_2$.
\item {\bf Outer Loop : }For $i$ from $1$ to $n$, do the following:
		\begin{itemize}
		 \item {\bf Inner Loop : }For $j$ going from $1$ to $\epsilon k$, do the following:
		 	\begin{enumerate}
			\item If all the columns of $\overline {\eval_{i}}$ have been spanned by the columns in $\cal M$, then do nothing
			\item Else pick a uniformly random column $C$ of $\overline {\eval_{i}}$ that has not been spanned by the columns of $\cal M$,
				 and pick a uniformly random element $b$ of $\F_2$.  	
			\item Set ${\cal M} = {\cal M} \| C $ (appending $C$ as a new column of {\cal M}) and set ${\cal B} = {\cal B} \| b$ (appending $b$ to the vector ${\cal B}$.  
			\end{enumerate}
		\item Set $A_i = \{[f] \mid [f] \times {\cal M} = {\cal B}; [f] \in \F_2^{kd}\}$  
		\end{itemize}
\item Let $S_0 = \{x_{i,j} \mid  j \neq f(i) ~\forall ~ [f] \in A_n\}$. Set all the variables $x_{i,j}\in S_0$ to $0$.
\end{enumerate}
\vspace{5mm}

The above random restriction procedure imposes at most $\epsilon k \times n$ independent $\F_2$-affine constraints on $[f]$. Each constraint restricts the space of possible $[f]$ by codimension $1$. Thus in the end $A_n$ is an affine subspace of $\F_2^{kd}$ of codimension at most $\epsilon k \times n$.  This immediately implies the claim below which shows that the size of $A_n$ is large. This in turn will imply that many of the monomials in $NW_d$ will survive after the random restriction. 

\begin{claim}\label{claim:anlarge}
$|A_n| \geq n^d/2^{\epsilon kn} = n^{d-\epsilon n}. $
\end{claim}
\begin{proof}
The main observation is that each time we are in Step (b) of the inner loop, we impose an {\it independent} $\F_2$-affine constraint on the possible choices of $[f]$. Thus the space of possible $[f]$ reduces by codimension exactly $1$. Thus we never impose conflicting constraints on $[f]$ and we ensure that at each step the number of $[f]$ satisfying all constraints is large. 
\end{proof}

\subsection{Effect of random restriction on $NW_d$}
Let $S_0$ be the set of variables output by the random restriction procedure $R_\epsilon$.  Let $R_\epsilon(NW_d)$ be the polynomial obtained from $NW_d$ after setting the variables in $S_0$ to $0$. In this section we will show that $R_\epsilon(NW_d)$ continues to remain hard in some sense. More precisely,  we will show that for any $S_0$ output by the $R_\epsilon$, and for $r<d$, a lot of distinct $r^{th}$ order partial derivatives of $R_\epsilon(NW_d)$ are non zero. 

Let $r< d-1$. Let $S \subset [n]$ be a set of size $r$. Let $T_S = \{\prod_{i \in S} x_{i,j_i} \mid  (j_i)_{i \in S} \in [n]^r \}$ be a set of $n^r$ monomials. We will consider partial derivatives of $NW_d$ with respect to monomials in $T_S$ for some choice of $S$. 

\begin{lem}[Random restriction on $NW_d$]~\label{lem:many-derivatives}
For every $\epsilon > 0$, and every set $S_0$ output by the random restriction procedure $R_\epsilon$, there is a set $S \subset [n]$ of size $r$ such that at least  $n^{r(1-\epsilon n/d)}$ monomials in $T_S$ are  such that the partial derivative of $R_\epsilon(NW_d)$ with respect to each of these monomials is nonzero and distinct.
\end{lem}

\begin{proof}
Observe that for any polynomial of degree at most $d-1$, its evaluation at some $d$ distinct points uniquely determines it. Let $S_i \in [n]$ be the set $\{(i-1)r +1, (i-1)r +2, \ldots, ir\}$. We will consider the set of evaluations of $f$ such that $[f] \in A_n$ at points of the set $S_i$ for various $i$. We will show that for some choice of $i$, the number of distinct sets of evaluations in $S_i$ as $[f]$ ranges in $A_n$ is large. 
Let $m_i$ be the number of distinct $r$-tuples of evaluations on $S_i$  as $[f]$ varies in $A_n$. Thus the total number of distinct $d$-tuples of evaluations on $[d]$ as  $[f]$ varies in $A_n$ is at most $\prod_{i=1}^{d/r} m_i$. However each $d$-tuple of evaluations on $[d]$ uniquely identifies $[f] \in A_n$. Thus $|A_n| \leq \prod_{i=1}^{d/r} m_i$. However by Claim~\ref{claim:anlarge} we know that $|A_n| \geq n^d/2^{\epsilon kn} = n^{d-\epsilon n}$. 
Thus there exists $i \leq d/r$ such that $m_i \geq n^{r(1 - \epsilon n/d)}$. Thus there are $n^{r(1 - \epsilon n/d)}$ monomials in $T_{S_i}$ each of which is consistent with some polynomial $f$ such that $[f] \in A_n$. Thus for each such monomial, there exists a monomial in $R_\epsilon(NW_d)$ extending it, and hence the corresponding partial derivative is nonzero. From Remark~\ref{rmk2} it follows that each of these partial derivative is distinct.
\end{proof}

\subsection{Effect of random restriction on $\spsp$ circuit}
Let $C$ be a homogeneous $\spsp$ circuit of size at most $n^{\rho\log\log n}$ for some very small constant $\rho$ that we will choose later. We will use $R_\epsilon(C)$ to refer to the $\spsp$ circuit obtained from $C$ after setting the variables in $S_0$ to $0$. This operation simply eliminates those monomials computed at the bottom later of $C$ which contain at least one variable which is set to $0$. Observe that homogeneity is preserved in this process. We will now show that with very high probability over the random restrictions, no product gate in $C$ at the bottom layer which takes more than $\Omega(\log n)$ distinct variables as input survives.

\begin{lem}[Random restriction on $\spsp$ circuit]~\label{lem:ckt-simplifies}
Let $\epsilon > 0$ and $\beta > 0$ be constants. Then there exists $\rho > 0$ such that if $C$ is a $\spsp$ circuit of size at most $n^{\rho\log\log n}$, then with probability $> 9/10$, all the monomials computed at the bottom layer which have support at least $\beta\log n$ have some variable set to $0$ by $R_\epsilon$. 
\end{lem}

Before we prove this lemma, we will first prove some simple results about affine subspaces and the probabilities of variables surviving the random restriction process.

\begin{prop}~\label{lem:random-subspace-inclusion}
Let $V$ and $W$ be fixed subspaces of $\F^k_2$ such that $W$ is a subspace of $V$. Let $U$ be a subspace of $V$ which is chosen uniformly at random among all subspaces of $V$ of dimension $\dim(U)$. Then, the probability that $W$ is a subspace of $U$ is at most $\prod_{j = 0}^{(\dim(W)-1)}\frac{2^{\dim(U)} - 2^j}{2^{\dim(V)} - 2^j} \leq 2^{-(\dim(V)-\dim(U))\dim(W)}$.
\end{prop}
\begin{proof}
Let us consider $Y$ to be a fixed subspace of dimension $\dim(U)$ of $V$. Now, let $A_U$ be an invertible linear transformation from $U$ to $Y$. Since, $U$ is chosen uniformly at random, so $A_U$ is also a uniformly random invertible matrix. Now, $W$ was a subspace of $U$ if and only if $A_UW$ is a subspace of $Y$. But since $A_U$ is chosen uniformly at random, so $A_UW$ is a uniformly random subspace of $\F^k_2$ of dimension $\dim(W)$. So, the desired probability is the same as the probability that for a fixed subspace $Y$ of dimension $\dim(U)$, a uniformly at random chosen subspace $W$ of dimension $\dim(W)$ lies in $Y$. Observe that sampling a uniformly random subspace can be done by greedily and uniformly at random sampling independent basis vectors for the subspace. Thus $W$ is contained in $Y$ if and only if all of the $\dim(W)$ linearly independent basis vectors chosen while randomly sampling $W$ lie in $Y$.  This quantity is at most $\prod_{j = 0}^{(\dim(W)-1)}\frac{2^{\dim(U)} - 2^j}{2^{\dim(V)} - 2^j}$.  Since, $\dim(U) \leq \dim(V)$, this probability is upper bounded by $2^{-(\dim(V)-\dim(U))\dim(W)}$.
\end{proof}

We will now visualize our variables to be arranged in an $n\times n$ {\it variable matrix}, where the $(i,j)^{th}$ entry of this matrix is the variable $x_{i,j}$. We say that a monomial {\it survives} the random restriction procedure given by $R_\epsilon$ if no variable in the monomial is set to zero. 

\begin{define}[Compact row]
We say that the $i^{th}$ row  in the variable matrix is {\it compact} if the columns of $\cal M$ sampled by the random restriction algorithm span every column of $\eval_{i}$. Thus $\cal M$ and $\cal B$ uniquely determine the value of $f(\alpha_i)$.  We say a row is non-compact otherwise. 
\end{define}

\begin{prop}~\label{lem:compact}
Suppose that the $i^{th}$ row of the variable matrix is compact. Then, for every $j \in \F_n$, the probability that a variable $x_{i,j}$ survives $R_\epsilon$ is at most $\frac{1}{n}$.
\end{prop}
\begin{proof}
The columns of $\cal M$ sampled by the random restriction algorithm span every column of $\eval_{i}$, so the value of $\cal B$ uniquely determines the value of $[f]\times \eval_{i}$. Moreover, since the columns of $\eval_{i}$ are linearly independent (since for every $j \in [n]$, there exists an $f$ such that $f(i) = j$) and $\cal B$ is chosen uniformly at random, so the value of $[f]\times \eval_{i}$ is a uniformly random element of $\F_2^k$. This implies that the value of $f(i)$ is uniquely determined and is a uniformly random element of $\F_n$. Thus the probability that $f(i) = j$ equals $1/n$, and the result follows. 
\end{proof}

\begin{prop}~\label{lem:noncompact}
Suppose that the $i^{th}$ row of the variable matrix is non-compact. Then, for every $j \in \{1, 2, \ldots, n\}$, the probability that $x_{i,j}$ survives is at most $\frac{1}{n^{\epsilon}}$. In fact this holds even after conditioning on any choice of $A_{i-1}$, which is the affine subspace $[f]$ is allowed to vary in after $i-1$ stages on the random restriction algorithm. 
\end{prop}
\begin{proof}
In the random restriction algorithm, since $i$ is a non-compact row, in stage $i$, we picked $\epsilon k$ independent columns of $\overline{\eval_{i}}$. At the end of stage $i-1$, $[f]$ was restricted to vary in some affine subspace $A_{i-1}$. Thus the possible values of $f(i)$ also varied in some affine subspace $V$.
At the end of stage $i$, $[f]$  was restricted to vary in some affine subspace of codimension $\epsilon k$ of $A_{i-1}$. This affine subspace was chosen by restricting the values of $f$ at $i$. Thus $[f(i)]$ was allowed to vary in a random affine subspace of codimension $\epsilon k$ in $V$. Call this subspace $U$. Thus the probability that $x_{i,j}$ survives is at most the probability that $j$ lies in the subspace $U$, which is at most $|U|/|V| = \frac{1}{n^{\epsilon}}$.

\end{proof}

We will now prove that any monomial which has a large support in any row of the variable matrix survives the random restriction procedure with only a very small probability.

\begin{lem}~\label{lem:flat-monomial-die}
Any monomial which has a support larger than $t$ in a row in the variable matrix survives $R_\epsilon$ with probability at most $\frac{1}{n^{\epsilon \log t}}$. 
\end{lem}
\begin{proof}
Let $\alpha$ be a monomial which has support $\geq t$ in row $i$ of the variable matrix.
Let $S = \{x_{i,j_1}, x_{i,j_2}, \ldots, x_{i,j_t}\}$ be any subset of the variables in this support of size $t$. For $t = 1$, the lemma trivially holds. Now,  if $t>1$, then if the row $i$ is compact then this monomial survives with probability $0$. So, now we will assume that row $i$ is non-compact. Since we identified $\F_n$ with $\F_2^k$, $\{j_i, j_2, \ldots, j_t\} \subset \F_2^k$. 
There must be $\log t$ of these elements that are linearly independent. Let this set of independent elements be $\beta_1, \beta_2,\ldots, \beta_{\log t}$. Thus $\alpha$ survives only if for each $j$, there is an $f$ such that $[f] \in A_n$ and $f(i) = \beta_j$. 

Recall that in the random restriction algorithm, in stage $i$, we picked $\epsilon k$ independent columns
of $\overline{\eval_{i}}$. At the end of stage $i-1$, $[f]$ was restricted to vary in some affine subspace $A_{i-1}$. Thus the possible values of $[f(i)]$ also varied in some affine subspace $V$.  If each of $\beta_1, \beta_2,\ldots, \beta_{\log t}$ were not contained in $V$ then $\alpha$ does not survive. Thus let us assume that $\beta_1, \beta_2,\ldots, \beta_{\log t} \in V$.

At the end of stage $i$, $[f]$  was restricted to vary in some affine subspace of codimension $\epsilon k$ of $A_{i-1}$. This affine subspace was chosen by restricting the values of $f$ at $i$. Thus $[f(i)]$ was allowed to vary in a random affine subspace of codimension $\epsilon k$ in $V$. Call this subspace $U$. Let $W$ be the subspace given by the span of $\beta_1, \beta_2,\ldots, \beta_{\log t}$. Then $\beta_1, \beta_2,\ldots, \beta_{\log t} \in U$ if and only if $W \subseteq U$. 
By Lemma~\ref{lem:random-subspace-inclusion}, the probability of this happening is at most $\frac{1}{n^{\epsilon \log t}}$.

\end{proof}

Now, let us consider a monomial which has a large number of variables from different rows. We will now estimate the probability that this monomial survives. 

\begin{lem}~\label{lem:tall-monomial-noncompact-die}
Let $t < d-1$. Any monomial which has support in $t$ non-compact rows survives $R_\epsilon$ with probability at most $\frac{1}{n^{\epsilon t}}$.
\end{lem}
\begin{proof}
Let $\alpha$ be a monomial which has at least one variable in each of $t$ distinct non compact rows, say $i_1, i_2, i_3, \ldots, i_t$. From Lemma~\ref{lem:noncompact}, we know that a variable in row $i_j$, $j\in [t]$, survives with probability at most $\frac{1}{n^{\epsilon}}$. In fact, conditioned on the variables in $i_1, i_2, \ldots, i_j$ surviving for any rows $i_1, i_2, \ldots, i_j$,  the probability that the variable in row $i_{j+1}$ survives is at most $\frac{1}{n^{\epsilon}}$. Hence, all of them survive with probability at most $\frac{1}{n^{\epsilon t}}$. 
\end{proof}


We will now show that monomials which have nonzero support in many compact rows survive with very low probability.
\begin{lem}~\label{lem:tall-monomial-compact-die}
Let $t < d-1$. Any monomial which has nonzero support in $t$ compact rows survives $R_\epsilon$ with probability at most $\frac{1}{n^{t}}$.
\end{lem}
\begin{proof}
Let $i_1, i_2, \ldots, i_t$ be some $t$ distinct compact rows. It is easy to see that the columns of the matrices $\eval_{i_1}, \eval_{i_2}, \ldots, \eval_{i_t}$ are all linearly independent, since $f$ can take all possible values at the points $i_1, i_2, \ldots, i_t$. Therefore, the probability that some variable survives in one of these rows is independent of the probability that some variable in another row survives.  From Lemma~\ref{lem:compact}, we know that any variable in any of these rows survives with probability at most $\frac{1}{n}$. From the above two observations, the probability that any monomial with support in these rows survives is at most $\frac{1}{n^t}$.
\end{proof}

Together, Lemma~\ref{lem:flat-monomial-die}, Lemma~\ref{lem:tall-monomial-noncompact-die} and Lemma~\ref{lem:tall-monomial-compact-die} show that any monomial with large support survives only with a very small probability, which completes the proof of Lemma~\ref{lem:ckt-simplifies}. We formally prove this below.  
\\


\noindent{\it Proof of Lemma~\ref{lem:ckt-simplifies}:}
From Lemma~\ref{lem:flat-monomial-die}, we know that any monomial which  has at least $\frac{\frac{\beta}{100}\log n}{\log\log n}$ variables in any row survives with probability at most $\frac{1}{n^{\epsilon(\log {\frac{\beta}{100}} + 0.9\log\log n)}}$ (for $n$ large enough). Hence, for any circuit of size at most $n^{\rho\log\log n}$, where $\rho < \epsilon/2$, by the union bound, with high probability none of the  monomials which  has at least $\frac{\frac{\beta}{100}\log n}{\log\log n}$ variables in any row survives.  

Similarly, by Lemma~\ref{lem:tall-monomial-noncompact-die}, a monomial with nonzero support in at least ${\log\log n}$ non-compact rows survives with probability at most $\frac{1}{n^{\epsilon \log\log n}}$. Hence, for circuits of size $n^{\rho\log\log n}$, where $\rho < \epsilon/2$, with high probability none of these monomials survive. 

Similarly,  monomials with nonzero support in $\log\log n$ compact rows are eliminated with a very high probability if $\rho < 1/2$. Hence, at the end of any such random restriction process, with  probability very close to $1$, none of the surviving monomials has support larger than $\beta\log n$ if $\rho < \epsilon/2$.   
\qed

\section{Lower Bounds for $NW_d$}~\label{sec:lowerbounds}
In this section, we give a proof of our main theorem. We will heavily borrow from the proof of Theorem~\ref{thm:bounded-support-lower-bound} in Section~\ref{sec:small-support-lb}.  The following lemma provides a lower bound on the complexity of the $NW_d$ polynomial after restricting it via $R_{\epsilon}$. 

\begin{lem}~\label{lem:manyshifts-restricted-NW}
Let $\delta$ and $\epsilon$ be any constants such that $0<\epsilon, \delta < 1$. Let $d = \delta n$. Let $\ell, m, r$ be positive integers such that $n-r > d$, $r < d-1$, $m \leq N$, $m = \theta(N)$ and  for $\phi =  {\frac{N}{m}}$,  $r$ satisfies $ r \leq \frac{ (n-d)\log{ \phi} \pm O\left(\phi \frac{(n-d-r)^2}{N}\right)}{(1-\epsilon n/d)\log n + \log{ \phi}}$. Then, for every random restriction $R_{\epsilon}$, $$\dim(\langle\partial^{r} R_{\epsilon}(NW_d)\rangle_{(\ell, m)}) \geq 0.5n^{(1-\epsilon n/d)r}{N \choose m}{\ell-1 \choose m-1}$$
\end{lem}

\begin{proof}
The proof is analogous to the proof of Lemma~\ref{lem:manyshiftsNW} till the point we substitute the value of ${\cal M}^{[r]}$ in the calculations in the proof of Lemma~\ref{lem:manyshiftsNW}.  For $R_{\epsilon}(NW_d)$, the value to be substituted is now $n^{r(1-\epsilon n/d)}$ as shown in Lemma~\ref{lem:many-derivatives}. So, we know that 

$$\dim(\langle\partial^{r} R_{\epsilon}(NW_d)\rangle_{(\ell, m)}) \geq 0.5n^{(1-\epsilon n/d)r}{N \choose m}{\ell-1 \choose m-1}$$ 

as long the parameters satisfy 

\begin{equation}~\label{eqn1} 
 n^{r(1-\epsilon n/d)}\times \frac{(N-(n-d-r))!}{N!}\times \frac{m!}{(m-(n-d-r))!} \leq 1
\end{equation}  

Now, using the  approximation from Lemma~\ref{lem:approx},
\begin{eqnarray*}
\log\frac{N!}{(N-(n-d-r))!} &=& (n-d-r)\log N \pm O\left(\frac{(n-d-r)^2}{N}\right) \text{and} \\
\log\frac{m!}{(m-(n-d-r))!} &=& (n-d-r)\log m \pm O\left(\frac{(n-d-r)^2}{m}\right)
\end{eqnarray*}

Now, taking logarithms on both sides in Equation~\ref{eqn1} and substituting these approximations, we get 
$$(1-\epsilon n/d)r\log n \leq \log\left(\frac{N}{m}\right)^{n-d-r} \pm  O\left(\frac{(n-d-r)^2}{N} + \frac{(n-d-r)^2}{m}\right)$$ 

Substituting $m = \frac{N}{\phi}$ and noting that $\phi > 1$,  we require 
 $$(1-\epsilon n/d)r\log n \leq (n-d-r)\log{ \frac{N}{m}} \pm  O\left(\phi\frac{(n-d-r)^2}{N}\right)  $$
 and  
 $$r \leq \frac{ (n-d)\log{ \phi} \pm  O(\phi\frac{(n-d-r)^2}{N})}{(1-\epsilon n/d)\log n + \log{ \phi}}$$
 
Observe that for any constant $0< \delta < 1$ such that $d = \delta n$, $r$ can be chosen any constant times $\frac{n}{\log n}$ by choosing $\phi$ to be an appropriately large constant. 
So, for such a choice of $r$,  we get $$\dim(\langle\partial^{r} NW_d\rangle_{(\ell, m)}) \geq 0.5n^{(1-\epsilon n/d)r}{N \choose m}{\ell-1 \choose m-1}$$
\end{proof}

The following lemma proves a lower bound on the top fan-in of any homogeneous $\spsp^{\{\beta\log n\}}$ circuit for the $R_{\epsilon}(NW_d)$ polynomial for a constant $\beta$. The proof of the lemma is essentially the same as the proof of Theorem~\ref{thm:bounded-support-lower-bound}. 

\begin{lem}~\label{lem:bounded-support-lb-restricedNW}
Let $d = \delta n$ for any constant $\delta$ such that $0 < \delta < 1$. Then, there exist constants $\epsilon, \beta$ such that  any homogeneous $\spsp^{\{\beta\log n\}}$ circuit computing the $R_{\epsilon}(NW_d)$ polynomial for any random restriction $R_{\epsilon}$ has top fan-in is at least $2^{\Omega(n)}$.
\end{lem}
\begin{proof}
By comparing the complexities of the circuit and the polynomial as given by Corollary~\ref{cor:circuit-complexity-bound} and Lemma~\ref{lem:manyshiftsNW}, the top fan-in of the circuit must be at least 

\begin{equation*}~\label{eqn:top-fan-in-2}
\frac{0.5n^{(1-\epsilon n/d)r}{N\choose m}{\ell-1 \choose m-1}}{\text{poly}(nrs){n+r \choose r}{N \choose m+rs}{\ell+n-r \choose m+rs}}
\end{equation*}

This bound holds for any choice of positive integers $\ell, m, r$, a constant $\beta$ such that $s = \beta\log n$ which satisfy the constraints in the hypothesis of Corollary~\ref{cor:circuit-complexity-bound} and Lemma~\ref{lem:manyshifts-restricted-NW}. In other words, we want these parameters to satisfy 

\begin{itemize}
\item $m+rs \leq \frac{N}{2}$ 
\item $m+rs \leq \frac{\ell}{2}$
\item $n-r > d$
\item $r < d-1$
\item For $\phi = {\frac{N}{m}}$,  $ r \leq \frac{ (n-d)\log{ \phi} \pm O\left(\phi\frac{(n-d-r)^2}{N}\right)}{(1-\epsilon n/d)\log n + \log{ \phi}}$
\end{itemize}

In the rest of the proof, we will show that there exists a choice of these parameters such that we get a bound of $2^{\Omega(n)}$ from expression above. 
We will show the existence of such parameters satisfying the asymptotics $\ell = \theta(N)$, $r = \theta\left( \frac{n}{\log n} \right)$ and $s = \theta(\log n)$. In the rest of the proof, we will crucially use these asymptotic bounds for various approximations.



Let us now estimate this ratio term by term. We will invoke Lemma~\ref{lem:approx} for approximations. 

\begin{itemize}
\item $\frac{{N \choose m}}{{N \choose m+rs}} = \frac{(N-m-rs)!(m+rs)!}{(N-m)!m!} = (\frac{m}{N-m})^{rs}$ upto some constant factors, as long as $(rs)^2 = \theta (N) = \theta(m)$. 
\item $\frac{{\ell-1 \choose m-1}}{{\ell+n-r \choose m+rs}} = {\frac{(\ell-1)!}{(m-1)!(\ell-m)!} \times \frac{(m+rs)!(\ell-m+n-r-rs)!}{(\ell+n-r)!}}$. Lets now pair up things we know how to approximate within constant factors. $\frac{{\ell-1 \choose m-1}}{{\ell+n-r \choose m+rs}} = \frac{(\ell-1)!}{(\ell+n-r)} \times \frac{(m+rs)!}{(m-1)!} \times \frac{(\ell-m+n-r-rs)!}{(\ell-m)!} = \text{poly(n)} \times {\frac{1}{\ell^{n-r}}}\times m^{rs} \times {\frac{(\ell-m)^{n-r}}{(\ell-m)^{rs}}}$. This simplifies to $\text{poly(n)} \times {\left(\frac{m}{\ell-m}\right)}^{rs} \times {\left(\frac{\ell-m}{\ell}\right)}^{n-r}$.
\item $\frac{n^{(1-\epsilon n/d)r}}{{n+r \choose r}} \geq \frac{n^{(1-\epsilon n/d)r}}{{\left(\frac{2(n+r)}{r}\right)}^{r}}$. We just used Stirling's approximation here. 
\end{itemize}


In the asymptotic range of our parameters, the approximations above imply that the top fan-in, up to polynomial factors is at least 
$${\left(\frac{r}{3}\right)}^r\times{\left(\frac{m}{\ell-m}\right)}^{rs} \times {\left(\frac{\ell-m}{\ell}\right)}^{n-r} \times \frac{1}{n^{(\epsilon n/d)r}} \times \left(\frac{m}{N-m}\right)^{rs}$$

Simplifying further, this is at least

$$2^{\Omega(r\log r - rs\log\frac{\ell-m}{m} - (n-r)\log\frac{\ell}{\ell-m} - (\epsilon n/d)r\log n - rs\log\frac{N-m}{m})}$$



We will set $m$ and $\ell$ to be $\theta(N)$ and $r$ to be $\theta(\frac{n}{\log n})$. The constants have to be chosen carefully in order to satisfy the constraints. We will choose constants $\alpha, \beta$ and $\eta$ such that $s = \beta\log n$, $r = \alpha\cdot n/\log n$ and $m= \eta \ell$. First let us choose $\epsilon$ to be a very small positive constant such that $\epsilon n/d = \epsilon/\delta << 0.1$ First choose $\eta$ to be any small constant $> 0$ (for instance $\eta= 1/4$). 
Now,  choose $\alpha$ to be a constant much much larger than $\log\frac{1}{1-\eta}$ and $\epsilon/\delta$. This makes sure that $r\log r$ dominates $(n-r)\log\frac{\ell}{\ell-m}$ and $(\epsilon n/d)r\log n$. Recall that $\alpha$ can be chosen to be any large constant by choosing $\phi$ to be appropriately large constant (by the constraint between $r$ and $\phi$ in the fifth bullet). Notice that this sets $m$ to be a small constant factor of $N$. Fix these choices of $\eta$ and $\alpha$. Now, we choose the term $\beta$ to be a small constant such that $rs\log\frac{1-\eta}{\eta}$ and $rs\log\frac{N-m}{m}$is much less than $r\log r$. Observe that this choice of parameters satisfies all the constraints imposed in the calculations above. Hence, the top fan-in must be at least $2^{\Omega(r\log r)} = 2^{\Omega(n)}$.  

\end{proof}

We now have all the ingredients to prove our main theorem. 
\begin{thm}~\label{thm:main3}
Let  $d = \delta n$ for any constant $\delta$ such that $0 < \delta < 1$. Any homogeneous $\spsp$ circuit computing the $NW_d$ must have size at least $n^{\Omega{(\log\log n)}}$.
\end{thm}
\begin{proof}
For every value of $\delta$, such that $0 < \delta < 1$, choose the parameters $\epsilon = \tilde{\epsilon}, \beta = \tilde{\beta}$ such that Lemma~\ref{lem:bounded-support-lb-restricedNW} is true for $\tilde{d} = \delta n$. Now, let us choose a constant $\rho = \tilde{\rho}$ such that Lemma~\ref{lem:ckt-simplifies} holds. Now, let $C$ be a homogeneous $\spsp$ circuit computing  the $NW_{\tilde{d}}$ polynomial. If the number of bottom product gates of $C$ was at least $n^{\tilde{\rho}{\log\log n}}$, then $C$ has large size and we are done.  Else, let us now apply a random restriction $R_{\epsilon}$ to the circuit. By the choice of parameters, Lemma~\ref{lem:ckt-simplifies} holds and so with probability $0.9$ every bottom product gate in $C$ with support larger than $\tilde{\beta}\log n$ is set to zero. After a restriction, the circuit computes $R_{\tilde{\epsilon}}(NW_{\tilde{d}})$. So, now we are in the case when we have a small support homogeneous circuit of depth four computing some random restriction of the $NW_{\tilde{d}}$ polynomial and then, by Lemma~\ref{lem:bounded-support-lb-restricedNW} above, the top fan-in of $R_{\tilde{\epsilon}}(C)$ must be at least $2^{\Omega(n)}$. Hence, any homogeneous $\spsp$ circuit computing $NW_{\tilde{d}}$ must have size at least $n^{\Omega(\log\log n)}$.
\end{proof}


\section{Open Problems}~\label{sec:conclusion}
The main question left open by this work is to prove much stronger, possibly exponential lower bounds for homogeneous $\spsp$ circuits. Given the earlier related works and the results of this paper, this question might be well within reach. It would be also very interesting to understand the limits of the new complexity measure of bounded support shifted partial derivatives that is introduced in this paper (as well as other variants) and investigate if they can be used to prove lower bounds for other interesting classes of circuits.

\section*{Acknowledgments}
We would like to thank Mike Saks and Avi Wigderson for many helpful discussions and much encouragement. We are also thankful to Amey Bhangale, Ben Lund and Nitin Saurabh for carefully sitting through a presentation on an earlier draft of the proof.

\bibliography{refs}

\end{document}